\documentclass[12pt,a4paper]{article}

\usepackage{cmbright,fullpage}
\usepackage{amsmath}%
\usepackage{amssymb}%
\usepackage{epsfig}%
\usepackage{amsthm}%
\usepackage{dsfont}%
\usepackage{natbib}
\usepackage{url}%
\usepackage{mathrsfs}%
\usepackage{braket}%
\usepackage{accents}
\usepackage{booktabs}
\usepackage{eurosym}
\usepackage{float}
\usepackage{placeins}

\usepackage{amssymb,latexsym,amsmath,soul,mathptmx,mathrsfs,verbatim}     
\usepackage[amsthm,mathsfit]{libertinust1math}
\newtheorem{rem}{\bf{Remark}}

\advance\voffset by 0pt
\advance\hoffset by 0pt
\renewenvironment{proof}[1][Proof]{\noindent\textbf{#1.} }{\ \rule{0.5em}{0.5em}}

\theoremstyle{plain}%

\newtheorem{cor}{Corollary}

\newtheorem{proposition}{Proposition}

\theoremstyle{definition}%
\newtheorem{definition}{Definition}

\renewenvironment{appendix}%
{\setcounter{section}{0}%
\numberwithin{equation}{section}%
}%
{}%
\begin{document}

  \title{Does the Market Anticipate? Can it? Should it? \thanks{The author thanks Prof. J. Thijssen, Dept. of Mathematics, University of York, for his patient advice.}
}
\author{KANGDA KEN WREN$^\mathbf{1}$\\
$^\mathbf{1}$k.k.wren@york.ac.uk\\Dept. of Mathematics, University of York}
\maketitle
\begin{abstract}
A nuance to 'no arbitrage' 
is studied: it can be suboptimal to act upon an arbitrage immediately; in such cases optimised trading can suppress the anticipation of predictable risk-outcomes, creating an apparent Status Quo Bias. This is shown for asset-pricing under model- or event-risk. Unlike standard treatments, we allow pre-horizon revelations of {risk-outcomes}; the associated technical challenges are addressed through results from the weak viability and (in)side-information literature. The conflict between 'anticipation', 'no arbitrage' and 'information efficiency', and the competition between the 'signal-to-noise of data' and the 'inference-free risk-premium drift', are exposed.
\end{abstract}
\newpage
\begin{center}
    {\huge Conflict-of-Interest Disclosure Statement}
\end{center}
\textbf{Author K.K.Wren}\\
I have nothing to disclose.
\newpage
\section
{Introduction}
Of the various notions of 'no arbitrage', the benchmark is NFLVR ('No Free Lunch with Vanishing Risk'), the basis for FTAP (Fundamental Theorem of Asset Pricing). The term 'information efficiency' commonly refers to how well prices incorporate 'the known' and 'the unknown'. To be clear, we call the second aspect 'anticipation', as distinct from the first, which will be referred to here as 'efficiency' exclusively.

For concreteness, our setting is continuous-time asset-pricing under some model- or event-risk \(B\) whose outcome \(b\) or \(\overline{b}\) 
is disclosed at a \emph{finite} point in time (e.g. war? recession? new CEO
? 'change' or 'no change'?). In our context, a 'model' is a probability law from a set of alternatives, and 'risk', the real or perceived randomness of an impactful outcome. 

\subsection
{The Apparent Dichotomy of Viability and Efficiency}
Call a market \emph{\(B\)-unsure} if it is NFLVR given public data, and \emph{\(\{B=b\}\)-sure} if it is NFLVR under the \emph{\(\{B=b\}\)-enlarged} dataflow. Consider the viability of 1) 
 a \(B\)-unsure market to the \emph{\(b\)-insiders}, who are \(\{B=b\}\)-sure
; 2) 
a \(\{B=b\}\)-sure market to the \emph{\(b\)-outsiders}, who are \(B\)-unsure
.

In case-1, the \(B\)-unsure market offers \(b\)-insiders only {unscalable} arbitrage {with no sure-payoffs pre-disclosure (\cite{sideinfo1}, \cite{ankimk}, \cite{sideinfo2}): there is no NFLVR-viability overall, only NUPBR-viability\footnote{
"No Unbounded Profit with Bounded Risk" is the minimum condition on prices for maximum-utility problems to be well-posed (\cite{kk}, \cite{fonz1}, \cite{fonz2}, \cite{fonz3}).}
.

In case-2, any \(\{B=b\}\)-sure market, to \(b\)-outsiders, has 
 a \emph{non-equivalent absolutely continuous martingale measure} (nACMM), a market-type studied by \cite{DelbS} despite stating its 'obvious unviability': all \(\{B\not=b\}\)-contingent rights, worthless in such a market, face infinite demand from the \(b\)-outsiders.
 
Hidden above is a seeming dichotomy: viability favours \(b\)-outsiders, 
who seek cheap \(\{B\not=b\}\)-\emph{rights}, unlike \(b\)-insiders, who issue \(\{B\not=b\}\)-\emph{obligations} and manage \(\{B\not=b\}\)-\emph{liabilities} ('margin calls') till \(B\)-resolution. As such, real-world markets, having to be viable against \emph{any sceptics of any fact}, would fail 'efficiency', or would they?

\subsection
{Asset-Pricing under Model- or Event-Risk}

Making do with 'working models' to be refined in due course is a common necessity in practice. 
Attempts to embed this into the Rational Expectation framework of asset-pricing theories include \emph{Parameter Learning} (e.g. \cite{GTim}, \cite{CD}), \emph{Robust Control} (e.g. \cite{hs1}, \cite{hs2}), and \emph{Ambiguity Aversion} (e.g. \cite{ep95}, 
\cite{ep02}, \cite{miao}, \cite{hndbk}).

While fruitful in many ways, the impact of these on mainstream theory and practice has been limited, arguably, due to complexity and a lack of connection to day-to-day investment challenges: these works, due to certain technical complications, tend to 'shun' Risk-Neutral Equivalent (RNE) asset-pricing
\footnote{
Ambiguity theories initially struggled with FTAP (\cite{ep95}), but their modern forms (e.g. \cite{ep02}, \cite{miao}), unified under Smooth Ambiguity, have no such issues (Proposition 3.6 of \cite{knoarb}). They need not, yet they do, avoid RNE pricing.} and exclude risks that resolve in finite times, sidestepping some highly relevant theoretical and practical questions.

\subsection
{Overview of Finding}
We apply RNE pricing after adding to the standard information structure risk-specific data that enable \emph{pre-horizon risk-resolution}.
The resulting price dynamic (Section \ref{viablepx}) shows that it can be \emph{suboptimal} to position immediately for 
a mispricing or arbitrage (Section \ref{noantpn}). As such, optimal trading can inhibit the anticipation of predictable outcomes in market prices, and the effect is reinforced under Ambiguity, by Bewley’s Inertia Axiom (\cite{Bew}). Through these the 'viability-efficiency dichotomy' can be averted, but at the expense of 'anticipation'. The theoretical and practical implications of the resulting Status Quo Bias, the potential countermeasures, as well as the importance of finding 'the right balance', are discussed at the end (Remark \ref{stqb}
)
.

 \subsection
 {Basic Background Assumptions}
We work in a frictionless, continuous and complete-market setting. The core stochastic basis is Wiener, based on classical Wiener spaces, essentially \(C_0[0,1]\) as well as \(C_0[0,1)\) (identified with \(C_0[0,\infty)\))
; the latter provides an explicit representation of \emph{finite-time risk-resolution}.

Asset-pricing under the binary \(B\)-risk is rational with respect to the \emph{reference beliefs}: two \(B\)-sure laws of data and one law of the risk \(B\). The resulting \emph{inferential stochastic basis} is the natural basis for asset-pricing. Let the reference \(B\)-sure laws be accurate by default, or else equivalent to the objective counterparts. The \(B\)-risk law can be unknown (ambiguous) and so the \(B\)-inferential beliefs, evolving with data, can be entirely {subjective}.
 
\subsection
{Technical Issues in Learning New Facts in Finite Time} 
In standard asset-pricing, data comprise only the direct inputs to 'worth', \emph{firmed-up benefits}\footnote{Classically dividend/interest, 'benefits' can be anything deemed 'valuable', widening the notion of 'worth'.}. Such a setup is \emph{irreducible}: future benefits are inferred from past ones. The model- or event-risks therein cannot resolve in finite time, so the resulting markets cannot eliminate \emph{a priori} ignorance or fallacies. Such learning is backward-looking and never-ending.

The above is diametric to the real world, where many types of data are deemed relevant and predictive, 
and risks resolve pre-horizon, even at known dates. This 'truth will out' backdrop matters, given the viability 
of essentially arbitrary prices otherwise (Remark \ref{2weak2strong} later). 
    \subsubsection*
    {Learning seems to make the pricing of resolvable risks intractable} Risks resolvable pre-horizon are priced by arbitrage, not the economics of utility/preference. Yet, it is hard to specify the price-dynamic enough for arbitrage pricing without the latter. Moreover, learning makes RNE probabilities path-dependent and RNE measures intractable, even in simple cases (e.g. \cite{GTim}). We show that these issues diminish in continuous time or under compound data. Explicit price dynamics emerge, with an intuitive and familiar price-of-risk; the RNE measure can be read off, at least up to easily absorbed well-behaved drifts.


\subsubsection*
{FTAP-viability can be impossible for processes that succeed in learning}
Under model- or event-risk, the inferential basis differ from the \emph{physically realised basis}. This is clear for risks concerning fact, hypothesis, parameter value (or \emph{rare-event}), etc., whose 'randomness' is (largely) subjective so that (mostly) only one outcome is realisable. Moreover, prices that are FTAP-viable in the inferential basis cannot be so in the true-outcome basis, due to divergence into any risk-resolution. Questions of 'actual viability' 
arise, as does the joint-hypothesis problem ("is an apparent anomaly due to market bias or flawed theories of market-price-of-risk?"; see \cite{jhp}). Such issues are moot 
in the absence of pre-horizon risk-resolution; in our approach they demand attention. 

We show that the more robust and general NUPBR-viability 
naturally arises in our setting. Moreover, under heterogeneous beliefs about model- or event-risk, a market with nACMM is found to emerge; 
it persists despite the arbitrage on offer to any inferential traders, 
thanks to such arbitrage being suboptimal to act upon until risk-resolution is imminent. 



\section{Setup}\label{setup}

Let all values be logarithmic, and time- and risk-free discount exogenous and set to nil
. Let horizon/maturity \(T\) be finite but far off: \(1\ll T<\infty\).

\subsection{A Classical Starting Point: the Simplest Standard Model}\label{assbasic}

Consider some FTAP-viable price process \(\{S^{b}_t\}\) of asset-'\(b\)' that offers a \emph{cumulative} firmed-up benefit stream \(\{Z_{t}(b)\}\), \(Z_{0}(b)=0\). 
Let both be driven by a standard Wiener noise \(\{dw_t\}\)
from the standard Wiener basis \((\Omega:=C_0[0,\infty),\{\mathcal{F}_t\};\text{ }w)\). For clarity on the main ideas, let benefit drift \(r_b\ge0\) and volatility \(\sigma>0\), along with \emph{risk-premium drift} \(\check{r}_b\ge0\), be constant.

The basic processes thus read as follows: given \(Z_t({b})=z_t\), \(t\le T\),\begin{align}\label{z0}
&dZ_t(b)={r}_bdt+\sigma dw_t
,\\
\label{S0}
&S^b_t=Y^b_t-RP^b(t),\text{ with }Y^{b}_t
=
z_t+r_b(T-t),\\
&\label{RP0}RP^{b}({t}):=Y^b_t-S^b_t=\check{r}_{b}(T-t),
\end{align}where \(Y^{b}_{t}\) is the {benefit expectation} given law (\ref{z0}), and \(RP^{b}(t)\), a {risk-premium}. 
It is convenient to extend the above beyond \(T\): on \([0,\infty)\ni t\) define \emph{base measure} \(W\) via \(dW|_t:=\sigma dw_t\) and the law \(W_{b}\) of benefit via \(dW_{b}|_t:={r}_{b}dt+dW|_t\). Both \(\{Y^b_t\}\) and \(\{S^b_t\}\) are {stopped} processes: \(Y^{b}_{t\ge T}\equiv S^{b}_{t\ge T}\equiv Z_T(b)\), so \(RP^b(\cdot)|_{t\ge T}\equiv0\); their martingale measures are respectively \(W_{b}|_T\) and \(\check{W}^T_{b}\sim W_{b}|_T\), with \(d\check{W}^T_{b}|_t=-\check{r}_{b}dt+dW_b|_t\), \(t\in[0,T]\). Note 
\(dS^b_t=\check{r}_bdt+dY^b_t\) and \(dY^b_t=dW|_t\).

 Classical asset-pricing (under independent increments) rests on two differentiable functions of time: expected growth \(y^b_T({t}):=
 {E}_b[Z_T(b)-z_{t}]\) and its RNE version \(\check{y}^b_T({t}):=
 \check{{E}}_{b}[Z_T(b)-z_{t}]\), with \(y^b_T({t})-\check{y}^b_T{(t)}=:RP^{b}({t})\). Here they read simply \(
 r_{b}(T-t)\) and \(
 (r_{b}-\check{r}_{b})(T-t)\) respectively, 
with a set-indicator \(\mathbb{I}_{\{\le T\}}(t)\) factor if defined on \([0,\infty)\).
\subsection{Adding Model- or Event-Risk the Standard Way}\label{addB}

Most model-risk studies focus on ignorance of the drift, which, unlike volatility, is hard to measure (e.g. \cite{mert}). At its most reduced, the task is one of choosing a 'true drift' from a set. Consider then model- or event-risk \(B\), of outcomes \(\mathcal{B}=\{b,\overline{b}\}\). 
To be concrete, set by default 
\(\mathcal{B}=\{+,-\}\), with \(r_\Delta:=r_+-r_->0\). Pricing under \(B\)-risk is guided by \(B\)-inference, based solely on firmed-up benefits, some \(\{Z_t(B)\}\) defined on \(({\Omega},\{\mathcal{F}_t\};\text{ }w)\) as before:\begin{equation}\label{z}
dZ_t(B)={r}_{B}dt+dW|_t=:dW_{B}|_t,
\end{equation}whose \(\{B=b\}\)-version is (\ref{z0}); any 
given \emph{\(B\)-sure FTAP-viable asset-price} \(\{S^B_t\}\) satisfies (\ref{S0}-\ref{RP0}), with RNE measure \(\check{W}^T_B\sim{W}_B|_T\). Wiener-drift tests (\cite{PShir}) are defined on the \emph{inferential basis} \(\big(\mathcal{B}\times{\Omega}
,\{\mathcal{F}_t\}
;\text{ }{\pi}^{B(\cdot)}_0{W}
_{B(\cdot)}(\cdot)\big)\), where \({\pi}^{B(\omega_B)}_0\) is an unconditional law of \(B(\omega_B)\in\mathcal{B}\), \(\omega_B:=(B,\omega)\in\mathcal{B}\times{\Omega}\). \emph{Regular tests} have \({W}_{+}\perp{W}_{-}\) and \({W}_{+}|_{t}\sim{W}_{-}|_{t}\), \(\forall t<\infty\), and resolve \(B\) continuously at \(T_{\mathcal{B}}:=\inf\big\{t|\mathcal{F}_t\ni\{B=\cdot\}\big\}=\infty\). See Appendix \ref{infps} for more. 

\begin{rem}\label{boundeddrift}
    Standard setting implies \({W}_{+}|_T\sim{W}_{-}|_T\) ((\ref{z})) and so no risk-resolution, hence the term '{hidden regimes}' for such risks. Inferred in arrears, never resolved, they are 'real' model-risks, not just queries about a classical process (e.g. "will (\ref{z0}) exceed \(z\)?")
    .
\end{rem}

\subsection{Risk-Specific Data and Scheduled Pre-Horizon Disclosure
}\label{Bdata}

We now add two new features to the standard structure. First, let there be 
data \(\{{D}_t(B)\}\) 
that is \emph{\(B\)-specific}: it is \emph{independent} of benefit \(\{Z_t(B)\}\) 
and is thus {irrelevant} to \(B\)-sure prices. If \(\{D_t(\pm)\}\) is defined on an usual Wiener basis \((\Omega
, \{\mathcal{F}^D_t\};\text{ }W^D_\pm)\), \({W}^D_{+}\perp{W}^D_{-}\), \({W}^D_{+}|_t\sim{W}^D_{-}|_t\), \(\forall t<\infty\), drift-testing proceeds as usual, with compound data \(\{\mathbf{D}_t\}:=\{(D_t,Z_t)\}\), filtration \(\{\mathbf{F}_t\}:=\{\mathcal{F}^D_t\otimes\mathcal{F}_t\}\)
 and so 'excessive' variabilities \emph{vs} testing based on \(\{Z_t\}\) alone.
 
 Second, let there be \(\{\mathcal{F}^D_t\}\)-predictable \(B\)-disclosure at some  \(t_D\ll T\); make it a \emph{known time} and set \(t_D=1\), with no loss of generality \emph{vs} any \(\{\mathcal{F}^D_t\}\)-stopping-time\footnote{Under \(t_D\ll T\) the value impact of \(B\)-outcomes has negligible erosion over \([0,t_D]\).} \(t_D\ll T\).
 
 The stochastic basis for the above can be found 
 by identifying \(C_0[0,\infty)\) and \(C_0[0,1)\) as normed metric spaces 
 (item-\ref{contim''}
 , Appendix \ref{infps}). Denote it by \(\big(C_0[0,1),\{\mathcal{F}^D_t\};\text{ }W^D_\pm\big)\), \(W^D_+\perp W^D_{-}\), \({W}^D_{+}|_{t}\sim{W}^D_{-}|_{t}\), \(\forall t<1\), with a slight abuse of notation. Over any \emph{pre-disclosure closed-interval} \([0, 1-\zeta]\text{, }\zeta\ll1\text{ arbitrary}\), let data \(\{D_t\}\) have the usual form in terms of a standard Wiener noise 
\(dw^D_t\) on \(C_0[0, 1-\zeta]\), with continuous \(\mathcal{F}^D_t\)-adapted drift \(r^D_{B,t}\) and volatility \(\sigma^D_t\):
\begin{equation}\label{bD}
    dD_t(B)=r^D_{B,t}dt+\sigma^D_tdw^D_t=:dW^D_B|_t,\text{ } t\in [0, 1-\zeta],
\end{equation}where \(r^D_{B,t}\) and \(\sigma^D_t\) diverge\footnote{Such dynamics may emerge also when the \(B\)-risk involved is an 'indicator pattern' (of various economic indices) to be declared 'on' or 'off' at \(t=1\), a typical scenario of '(in)side information' studies.
} as \(t\to1\); 
 set \(\{r^{D}_{\Delta,t}\}:=\{r^D_{+,t}\}-\{r^D_{-,t}\}>0\) for concreteness.

Our information structure adds modelling versatility
. For instance, it caters for a common real-world scenario: a market  'triggered' into a bout of model- or event-risk, whose economics dominate (routine ones) and whose truth can be learnt at a known date. 


\begin{rem}\label{anythingatall}
    Note that the \(B\)-risk and \(B\)-specific data may represent any, not necessarily sound, system of beliefs and data (e.g. astrological). Our compound structure with pre-horizon risk-resolution is a tractable way to capture irrational but self-consistent market aspects.
\end{rem}



\section{The RNE Formulation of Model- or Event-Risk Pricing}\label{viablepx}

By definition and the setup so far, we have:\begin{equation}\label{dY(B)}
    dY^b_t(B)=dZ_t(B)+dy^b_T(t)=sign(\overline{b})\mathbb{I}_{\{\overline{b}\}}(B)\cdot r_\Delta dt+dW|_t;
\end{equation}hence \(dY^B_t(B)\equiv dW|_t\). As such, 
 only drift-difference counts in inferential testing 
 .
 De-trending raw data makes \(\{\mathbf{D}_t(-)\}\) a martingale, gives \(\{\mathbf{D}_t(+)\}\) a drift of the form \(\{(r^{D}_{\Delta,t},r_\Delta)\}\), and leads to: \(dW_-|_t=\sigma dw_t\), 
\(dW^D_-|_t=\sigma^D_tdw^D_t\)
, \(y^-_T(t)=0\), \(Y^-_t=
z_t\). 

Further, as risk-premium compensates for volatility, which is unaffected by our \(B\)-risk
, we shall presume \emph{identical \(B\)-sure risk-premium}\footnote{This 'purist' view, just as the constant drift and volatility of asset-benefit, is non-essential, but the resulting lighter notations promote clarity on the core ideas.}: \(RP^+(t)=RP^-(t)=:\check{RP}(t)\), and therefore, by \({Y}^{\Delta}(t)-{S}^{\Delta}(t)\equiv RP^+(t)-RP^-(t)\), where \({Y}^{\Delta}(t):={{y}}^+_T(t)-{{y}}^-_T(t)>0\) and \({S}^{\Delta}(t):=\check{{y}}^+_T(t)-\check{{y}}^-_T(t)>0\), \emph{identical risk-impact} \({X}_{\Delta}(t)
>0\) on benefit expectations and prices.

\subsection{FTAP-Viable Pricing in the Inferential Basis}\label{coreprop}
Asset-pricing under \(B\)-risk is driven by the risk-beliefs \(\{\pi^B_t\}\) generated on the inferential basis \(\big(\mathcal{B}\times{\mathbf{\Omega}}_{1},\{\mathbf{F}_t\};\text{ }\pi^B_0\mathbf{W}_{B}\big)\), \(\mathbf{{\Omega}}_{1}:=C_0[0,1)\times C_0[0,1] \), \(
{\mathbf{W}}_{B}:=
W^D_B\times W_B|_{1}\). FTAP-viable asset-prices \(\{S_t\}:=\{Y_t\}-\{RP_t\}\), expected benefit \(\{Y_t\}\) less risk-premium \(\{RP_t\}\), must have a martingale measure equivalent to that of \(\{Y_t\}\), a well-known \(\{\mathbf{F}_t\}\)-martingale under \(\pi^B_0\mathbf{W}_{B}\).

\begin{rem}\label{2weak2strong}
    FTAP-viability is too broad for most modelling purposes: 
 without \(B\)-resolution, both \(B\)-sure 
 prices, \(z_{t}+\check{y}^{B}_T(t)\), are FTAP-viable on \([0,T]\) regardless of 'the truth about \(B\)'. Yet it is also too narrow: it is impossible if one of the \(B\)-outcomes is unphysical and so the \(B\)-risk, entirely subjective (e.g. 'Was there a takeover offer?'); see (\ref{resovlingls}). The first issue is addressed by pre-horizon \(B\)-disclosure, and the second, by NUPBR-viability (Section \ref{emo}).
\end{rem}

More specifically, the key to asset-pricing under \(B\)-risk is the risk-premium attributable to \(B\)-risk, to be denoted \(B\_RP_t\)
. Economic interpretability also requires \({S}_t\in
[{S}^-_t,{S}^+_t]\) and so the existence of \emph{pricing coefficients} \(0\le\{{A}^B_t\}\le1\) such that \({S}_t\equiv\sum_B{{S}^B_t}{A}^B_t=:\langle S\rangle^{A}_t\).

To summarise: with \(\sigma^{\pi}_t:={(\pi^+_t{\pi}^-_t)^{\frac{1}{2}}}\) and in shorthand \(\langle \cdot\rangle^{\pi}_t:=\sum_B{{(\cdot)}^B}{\pi}^B_t\),\begin{align}
    \label{yt}  
&{Y}_t\equiv\langle Y\rangle^{\pi}_t=Y^-_t+X_\Delta(t)\cdot\pi^+_t,
\\
\label{abn}&{S}_t\equiv\langle{S}\rangle^A_t={S}^-_t+{X}_{\Delta}(t)\cdot A^+_t
,\\
\label{brpn}&{B\_RP}_t:=
\langle S\rangle^{\pi}_t-{S}_t={RP}_t-\check{RP}(t)={X}_{\Delta}(t)({\pi}^+_t-A^+_t),\\
\label{k}&k^A_t:=\frac{{B\_RP}_t}{{X}_{\Delta}(t)\sigma^{\pi}_t}=\frac{{\pi}^+_t-A^+_t}
{\sigma^{\pi}_t};
\end{align}note the standard deviation \({X}_{\Delta}(t)\sigma^{\pi}_t\) of \(B\_RP_t\) and hence the \emph{price-of-\(B\)-risk} \(k^A_t\). Process \(\{k^A_t\}\) is \(\{\mathbf{F}_t\}\)-adapted, continuous, and non-negative 
(at least \emph{ex ante})
.

It is convenient to work with \emph{\(B\)-risk only pricing} \(\langle{Y}\rangle^A_t\), ignoring \(B\)-sure risk-premium \(\check{RP}(t)\),\begin{align}\label{S0A}&
\langle{Y}\rangle^A_t:=Y_t-B\_RP_t=
S_t+\check{RP}(t)
= Y^-_t+{X}_{\Delta}(t)\cdot A^+_t
;
\end{align}given tame \(\check{RP}(t)\) (as in most cases), the \(B\)-risk only price \(\langle Y\rangle^{A}_t\) is FTAP-viable \emph{iff.} the asset-price \(S_t
=\langle Y\rangle^{A}_t-\check{RP}(t)\) itself is. A comparison of \(\langle Y\rangle^{A}_t\) 
with expectation \(\langle Y\rangle^{\pi}_t\) ((\ref{yt})) suggests and proves (Appendix \ref{proofp3.1}) that 
\(\langle Y\rangle^{A}_t\) is FTAP-viable 
    \emph{iff.} pricing coefficients \(\{A^+_{t}\}\) amount to RNE \(B\)-risk beliefs: \(\{A^+_{t}\}=\{\hat{\pi}^+_t\}\sim\{\pi^+_t\}\). The claim is obvious should \(B\)-inference rely on \(\{D_t(B)\}\) only, 
    by the resulting independence of \(\{\pi^+_t(B)\}\) and \(\{Y^-_t(B)\}\). 
    \begin{proposition}\label{p3}
Under \(B\)-risk, with data \(\{\mathbf{D}_t\}\) governed by reference law \(\pi^B_0\mathbf{W}_B\) over \([0,1]\), giving rise to \(B\)-risk beliefs \(\{\pi^B_t\}\), 
any FTAP-viable asset-price \(\{S_t\}\) must have the form:\begin{align}\label{twintwin}
     &S_t=\langle{Y}\rangle_t^{\hat{\pi}}-\check{RP}(t),
 \end{align}where \(B\)-risk only asset-pricing \(\langle{Y}\rangle_t^{\hat{\pi}}\) 
 has RNE law of the form \(\hat{\pi}^B_0(\hat{W}^D_B\times{W}_B|_{1})\), with \(\hat{W}^D_B\sim{W}^D_B\), \(\hat{\pi}^B_0\sim{\pi}^B_0\) and RNE \(B\)-risk beliefs \(\{\hat{\pi}^B_t\}\sim\{{\pi}^B_t\}\).
 \end{proposition}
\subsubsection*{Economic Considerations and Canonical Risk-Pricing}
So far, beside the very obvious, we have not imposed any economic conditions. Most asset-pricing theories maximise expected utility, thereby making the price-of-\(B\)-risk \(k^A_t\) not only data-adapted but also a ({2-differentiable}) function \(\mathbf{k}^A\) of \(B\)-risk beliefs:\begin{equation}\label{kcon}
k^A_t\equiv \mathbf{k}^A(\pi^+_t).
\end{equation}Under FTAP-viability, this is a strong demand, thanks to Ito's Formula (Appendix \ref{coro1proof}):

\begin{cor} \label{coro0}
    FTAP-viable pricing (\ref{twintwin}) subject to (\ref{kcon}) must have \(B\)-risk only pricing \(\langle{Y}\rangle_t^{{\Pi}}\) based on RNE \(B\)-risk beliefs \({\Pi}^{B}_t\) under canonical RNE laws of the form \(\Pi^B_0\mathbf{W}_B\), \(\Pi^B_0\sim\pi^B_0\)
    . 
\end{cor}


\begin{cor} \label{coro1}Canonical 
price-of-\(B\)-risk \(k^{\Pi}_t:=\frac{{\pi}^+_t-\Pi^+_t}
{\sigma^{\pi}_t}\) 
has the following property:\begin{align}
\label{canon2}
        k^{\Pi}_t=(K^{\frac{1}{2}}-K^{-\frac{1}{2}})\sigma^{\Pi}_t,
        \end{align}where 
        \(K:={L^{\pi{+}}_0}/{L^{\Pi{+}}_0}
        ={L^{\pi{+}}_t}/{L^{\Pi{+}}_t} 
        \), with \(L^{\pi{+}}_{(\cdot)}:=\frac{\pi^+_{(\cdot)}}{
        \pi^-_{(\cdot)}}\) and \(L^{\Pi{+}}_{(\cdot)}:=\frac{\Pi^{+}_{(\cdot)}}{
        \Pi^-_{(\cdot)}}\), is a conserved quantity of the inference dynamic. Risk-pricing is positive iff. \(K>1\); unpriced \(B\)-risk has \(K=1\).
        \end{cor}

\begin{rem}\label{1-2}
      Given pricing \(\Pi^+_t\), the gain-loss ratio of bearing \(B\)-risk is \(\frac{\Pi^-_{t}}{\Pi^{+}_{t}}\equiv\frac{1}{L^{\Pi{+}}_t}\); it is the reference ratio \(\frac{1}{L^{\pi{+}}_t}\) scaled (up) by \(K\). Indeed, at peak \(B\)-risk (i.e. \(L^{\pi{+}}_t=1\)), the ratio is \(K\) itself. Buyers (sellers) want \(K\gg1\) (\(K\ll1\)). By (\ref{canon2}) the \(B\)-risk premium 
      has peak \(B\_RP^{max}:=\frac{\sqrt{K}-1}{\sqrt{K}+1}X_{\Delta}(t)\) at \({\Pi}^+_{max}=\frac{1
      }{\sqrt{K}+1}\) and so \({\pi}^+_{max}=\frac{\sqrt{K}
      }{\sqrt{K}+1}\), but at peak \(B\)-risk it has value \(B\_RP^{\frac{1}{2}}:=\frac{1}{2}\frac{K-1}{K+1}
      {X}_{\Delta}(t)\)
      . 
    We surmise that 'normal' market is such that 
      \({B\_RP^{max}}\approxeq{B\_RP^{\frac{1}{2}}}\ge0\) and so \(K\in[1,2)\).
      \end{rem}

Economic condition (\ref{kcon}) however is not assured, as our \(B\)-risk is resolvable pre-horizon and so open to arbitrage pricing. On the other hand, non-canonical pricing can be seen as canonical but under distorted \(B\)-sure beliefs (Appendix \ref{weakp'}).

\subsection{A Simplifying and Practically Relevant Concrete Case}\label{shortcut''}

Consider the risk of change from a given status quo; call outcome 'change' \(\{B=\mathbf{1}\}\), and 'no change', \(\{B=\mathbf{0}\}\). In our setting, their drift-difference then has the form \(r_{\Delta}(t)=\mathbb{I}_{\{>1\}}(t)\cdot r_\Delta\), so that \emph{post-disclosure} benefits depend on \(B\), but pre-disclosure, not:
\begin{align}
    \label{assmption3}
W_{\mathbf{0}}|_{t}=W_{\mathbf{1}}|_{t},\text{ i.e. }Z_{t}(\mathbf{0})=Z_{t}(\mathbf{1}),\text{ \(\forall t\le1\)}.\end{align}Note, as a result, \(dY^{\mathbf{0}}_t(B)=dW|_t
\), \(\forall t\le1\), regardless
, so that \(B\)-inference is driven by \(B\)-specific data \(\{D_t\}\) only, and \(B\)-risk impact, constant during inference:\begin{align}
\label{fixedimp}{X}_{\Delta}(t)
=
|{{y}}^{\mathbf{1}}_T(1)-{{y}}^{\mathbf{0}}_T(1)|\equiv
{{y}}^+_T(1)-{{y}}^-_T(1)=:X_{\Delta},\text{ \(\forall t\le1\)},
\end{align}where '\(\pm\)' labels the relative economics of one \(B\)-outcome \emph{vs} the other; 
let this be conveyed by indicator function \(sign(B)=\pm\) 
. 
Then, by (\ref{S0A}-\ref{twintwin}) and (\ref{rawdrift}) (Item-\ref{ckchng}, Appendix \ref{infps}), we have the following price dynamic over \((0,1)\ni t\). 
\begin{cor}\label{coro0'}
     Under (\ref{assmption3}-\ref{fixedimp}), asset-price (\ref{twintwin}) 
     has RNE law of the form \(\hat{\pi}^B_0(\hat{W}^D_B\times \check{W}^T_B|_{1})\). The price dynamic under \(B\)-risk reads: with respect to \(dw_t=\frac{dW|_t}{\sigma}\) and \(d\hat{w}^D_t:=\frac{d\hat{W}^D_-|_t}{\sigma^D_t}\),
\begin{equation}
       \label{fulldynamics}
dS
_{t}(B)=
\check{r}
dt+\sigma dw^{}_{t}+
X_{\Delta}\cdot(\sigma^{\hat{\pi}}_t)^2\big{[}(\mathbb{I}_{\{+\}}\circ sign -\hat{\pi}^{+}_t)(B)\cdot(\sigma^{lD}_t)^2dt+\sigma^{lD}_td\hat{w}^{D}_{t}\big{]},
\end{equation}driven by: 1) the risk-premium drift 
\(-d\check{RP}(t)
\), 2) baseline noise \(dY^{\mathbf{0}}_t\), and 3) RNE inference \(
X_{\Delta}\cdot d\hat{\pi}^+_t\) based on the signal-to-noise \(\sigma^{lD}_t:={r^{D}_{\Delta,t}}/{{\sigma^D_t}}\) of data \(\{D_t\}\).\end{cor}
In general, the RNE law, despite obvious existence, as \(dS_t=-d\check{RP(t)}+d\langle{Y}\rangle_t^{\hat{\pi}}\), is elusive: without (\ref{assmption3}-\ref{fixedimp}), the inferential \(d\hat{\pi}^+_t\) or \(d{\pi}^+_t\) terms 'interact' with the \(dZ_t\) terms
.


\begin{rem}\label{ubi}Under economic condition (\ref{kcon}), prices are canonical: \(\hat{\pi}^+_t=\Pi^+_t\) and \(d\hat{w}^D_t=dw^D_t\). The average price-dynamic with respect to \(B\)-outcomes given \(B\)-risk beliefs \(\pi^+_t\) reads:\begin{equation}\label{klasic}
\check{r}
dt+\sigma dw^{}_{t}+\boldsymbol\sigma_t\big[\frac{\boldsymbol\mu_t}{\boldsymbol\sigma_t}dt+dw^{D}_{t}\big],\end{equation} where \(\boldsymbol\sigma_t:=X_{\Delta}\cdot(\sigma^{{\Pi}}_t)^2\sigma^{lD}_t\), \(\boldsymbol\mu_t:=X_{\Delta}\cdot(\pi^+_t-{\Pi}^+_t)(\sigma^{{\Pi}}_t\sigma^{lD}_t)^2\), and so ((\ref{canon2})):\begin{equation}\label{klasic'}
    \frac{\boldsymbol\mu_t}{\boldsymbol\sigma_t}=
    \frac{K^{\frac{1}{2}}-K^{-\frac{1}{2}}}{X_{\Delta}}\cdot\frac{\sigma^{\pi}_t}{\sigma^{\Pi}_t}\cdot\boldsymbol\sigma_t=\frac{K-1}{X_{\Delta}(t)}\cdot\frac{\boldsymbol\sigma_t}{1+(K-1)\Pi_t}.
\end{equation}Under \((K-1)\)-expansion (Remark \ref{1-2}), we recover the ubiquitous price-vs-volatility formula:
\begin{align}\label{capmlike}
    &\frac{\boldsymbol\mu_t}{\boldsymbol\sigma_t}\approx\bigg(\frac{K-1}{X_{\Delta}}\bigg)\boldsymbol\sigma_t\text{, that is, }\boldsymbol\mu_t\approx\bigg(\frac{K-1}{X_{\Delta}}\bigg)(\boldsymbol\sigma_t)^2.
\end{align}Classically the scaling factor is risk-aversion. The above puts it in more tangible terms
. \end{rem}

\begin{rem}\label{views}
      Post-disclosure, on \((1,T]\), we recover under '{hidden disclosure}' the usual hidden-regime setup, and under '{private disclosure}' the usual insider setup with '{initially enlarged filtrations}' (e.g. \cite{ankimk}). In both cases, with only benefit data and no prospect of risk-resolution on \((1,T]\), economic and so canonical asset-pricing applies
      .
    \end{rem}  
\subsection{NUPBR-Viability at Risk-Resolution}\label{emo}

FTAP-viability has been shown in the inferential basis \(\big(\mathcal{B}\times{\mathbf{\Omega}}_{1},\{\mathbf{F}_t\};\text{ }\pi^B_0\mathbf{W}_B\big)\). Yet the \emph{true-outcome basis}, of the form \((\mathbf{{\Omega}}_{1},\{\mathbf{F}_t\};\text{ }\mathbf{W}_{b})\) if \(B=b\), permits \emph{no} FTAP-viability of the process that learns \(\{B=b\}\) 
(Item-\ref{convergent}, Appendix \ref{infps}). 

We show in Appendix \ref{weakp} that with respect to the physical outcome, price dynamic (\ref{fulldynamics}) is \emph{minimally viable} in the NUPBR sense, based on its compliance, as \(t\to 1\), to condition (\ref{finitereldrft}), Item-\ref{unique}, Appendix \ref{infps}. It excludes the 'worst' arbitrage, offering only unscalable classical arbitrage as \(B\)-disclosure approaches.
 
 \begin{rem}
     As such, FTAP-viability may seem surplus. See \cite{kk} and \cite{fonz2} for "markets without martingale measures". Nevertheless inferential FTAP-viability is a good starting point and a good fit for ex-ante rationality.
 \end{rem}
 
\section{The Dominance of Conviction and Non-Anticipation}\label{noantpn}
For the the convenience of this section, we re-state below asset-price dynamic (\ref{fulldynamics}) 
with respect to outcome \(\{B=\mathbf{1}\}\) ('change'): 
\begin{align}\label{ori}
    d{S}
    _{t}(B)=&\big[
    {\check{r}
    }dt+\sigma dw_{t}\big]+
    \\\label{domi}
+&sign(\mathbf{1})X_{\Delta}\cdot (\sigma^{\hat{\pi}}_t)^2\big[(\mathbb{I}_{\{\mathbf{1}\}}-\hat{\pi}^\mathbf{1}_{t})(B)({{\sigma}}^{lD}_t)^2dt+{\sigma}^{lD}_td{\hat{{w}}^D_t}\big].\end{align} It resembles a \(B\)-sure asset ((\ref{ori})) with 'abnormal excess' ((\ref{domi})). This excess 
stems from \(B\)-risk and its data. 
Only \(B\)-sure 
price dynamics deliver the target drifts. 

\subsection{A Market of Heterogeneous Beliefs: 'Change' \emph{vs} 'Status Quo'}\label{canonbiased}In the setting of Section \ref{Bdata}, pre-disclosure, consider two trader-types: the \(\mathbf{0}\)-traders, sure of status quo
, and the inferential \({\Pi}\)-traders, unsure. They are otherwise identical: \(\mathbf{0}\)-traders believe in law \(\mathbf{W}_{\mathbf{0}}
\) and share status quo fair-price \(S^{\mathbf{0}}_t\), while \({\Pi}\)-traders believe in \({\pi}^{B}_0 \mathbf{W}_{B}\) 
and share economic, canonical, fair-price \(\langle S\rangle^{{\Pi}} _t\equiv S^{\mathbf{0}}_t+sign(\mathbf{1})X_{\Delta}\cdot{\Pi}^\mathbf{1}_t
\). 

Let there be classical (\(B\)-sure) assets 
offering the target drift \(\check{r}dt\). 
Traders seek \emph{excess} (to beat \(\check{r}dt\)): so 
\(\mathbf{0}\)-traders bet against 'change'; \({\Pi}\)-traders, for, unless overpriced. At any \(t<1\), given a market price of \({S}^{\mathbf{0}}_t+sign(\mathbf{1})\Delta S_t\), all in log-value, with \(0\le{\Delta S_t}\ll1\),
\begin{itemize}\item the \(\mathbf{0}\)-traders sell whenever \(sign(\mathbf{1})\Delta S_t>0\), and buy otherwise;\item the \({\Pi}\)-traders buy whenever \(sign(\mathbf{1})({\Pi}^\mathbf{1}_t-\Delta S_t/X_{\Delta})\ge0\), and sell otherwise.
\end{itemize}

Market-clearing details notwithstanding, directional conclusions can be drawn. We start with the setup of \cite{BH}, (\citeyear{BH98}) and \cite{BHW}: the relative 'fitness' of trader-factions is measured by some weighted-average of the excess realised to date, \(U_{\mathbf{0},t}\) and \(U_{{\Pi},t}\), so that the ratio \(M_t\) between the \(\mathbf{0}\)- and \({\Pi}\)-factions obey the limit law of standard \emph{stochastic discrete choice models} 
(\cite{BH}):\begin{equation}\label{mratio}M_t:=\exp[\beta_t\cdot(U_{\mathbf{0},t}-U_{{\Pi},t}+C_t)]>0,\end{equation}where \(\beta_t>0\) is an 'intensity of choice', 
and \(C_t\ge0\), the cost of being \({\Pi}\)- \emph{vs} \(\mathbf{0}\)-traders. The clearing price \({S}^{het}_t\) of this heterogenous market at any \(t<1\), by (2.7) of \cite{BH98}, is a \(M_t\)-weighted average of the corresponding homogeneous-market prices:\begin{align}\label{Shetro}
&{S}^{het}_t=\langle S\rangle^{\frac{{\Pi}^\mathbf{1}_t}{1+M_t}}_t\equiv
{S}^{\mathbf{0}}_t+sign(\mathbf{1})X_{\Delta}\cdot\frac{{\Pi}^\mathbf{1}_t}{1+M_t}.
\end{align}To be viable to the given participant beliefs, the pricing coefficients 
must amount to some RNE \(B\)-risk belief process (Proposition \ref{p3}): 
\(\hat{\pi}^\mathbf{1}_{t}=
\frac{{\Pi}^\mathbf{1}_{t}}{1+M_{t}}<{\Pi}^\mathbf{1}_{t}\) 
and \(S^{het}_{t}=
{S}^{\mathbf{0}}_t+sign(\mathbf{1})X_{\Delta}\cdot\hat{\pi}^{\mathbf{1}}_t\). 

\subsection{'Status Quo' Wins}\label{domisq}
 A \emph{steady state} of this market must have a stable trader-mix and no average trading excess:
 \begin{enumerate}
    \item \emph{The \({\Pi}\)-trader price \(\langle S\rangle ^{{\Pi}}_t
    \) cannot be a steady state.} Its dynamic (\ref{ori}-\ref{domi}) is canonical: \(d\hat{w}^D_t=dw^D_t\) and \(\hat{\pi}^{\mathbf{1}}_t={\Pi}^{\mathbf{1}}_t\). Over \((t,t+\Delta t<1]\), \(\Delta t\ll1\), 
    given any time-\(t\) market price of \({S}^{\mathbf{0}}_t+sign(\mathbf{1})\Delta S_t\), \(0\le{\Delta S_t}\ll1
    \), 
    the average \(\mathbf{0}\)-trader excess, 
    to \(\mathcal{O}(\Delta t\cdot\Delta S_t)\), is:\footnote{
    Any deviation \(|\Delta S|\) from  \({S}^{\mathbf{0}}_t\) causes excess \(|\Delta S|\cdot\check{r}\Delta t\), on top of any value-change beyond \(\check{r}\Delta t\) during \(\Delta t\).}
    \begin{itemize}
    \item
    \(
    {\check{r}}\Delta t\cdot\Delta S_t>0,\text{ if \(B=\mathbf{0}\)}\);
    \item \(\big{[}{\check{r}}-({\sigma}^{lD}_t)^2\big{]}\Delta t\cdot\Delta S_t=:\big{[}1-RI_t\big{]}{\check{r}}\Delta t\cdot\Delta S_t,\text{ if } B=\mathbf{1}\);
    \end{itemize}
    the average \({\Pi}\)-trader excess, 
    to \(\mathcal{O}(\Delta t\cdot\Delta S_t)\), is:
\begin{itemize}\item\(
-(+)sign(\mathbf{1})
{\check{r}}\Delta t\cdot\Delta S_t\text{ for \(sign(\mathbf{1})({\Pi}^\mathbf{1}_t-\Delta S_t/X_\Delta)\ge0\text{ }(<0)\), if \(B=\mathbf{0}\);
}
\)\item\(
-(+)sign(\mathbf{1})\big{[}1-RI_t
\big{]}{\check{r}}\Delta t\cdot\Delta S_t\text{ for \(sign(\mathbf{1})({\Pi}^\mathbf{1}_t-\Delta S_t/X_\Delta)\ge0\text{ }(<0)\), if \(B=\mathbf{1}\)};\)\end{itemize}thus, regardless of \(B\), whenever \(RI_t:=\frac{({\sigma}^{lD}_t)^2}{{\check{r}}}<1\),  the \(\mathbf{0}\)-traders win and the \(\Pi\)-traders, if betting on 'change', lose
. 
As such, selection pressure (\ref{mratio}) diminishes the latter quickly, especially at high intensity of choice \(\beta_t\) or cost differential \(C_t\).

\begin{definition}\label{riilow}
    The \emph{relative information intensity} (RI) at any time \(t\) refers to the ratio \(\frac{({\sigma}^{lD}_t)^2}{{\check{r}}}=:RI_t\), and \emph{low RI}, to \(\{RI_t\le1\}\). 
\end{definition}   

\item \emph{The heterogenous-market clearing price \({S}^{het}_{t}=\langle S\rangle ^{\hat{\pi}}_t\) cannot be a steady state.} It has dynamic (\ref{ori}-\ref{domi}). The same argument and outcome scenarios as above apply. 


\item\emph{The \(\mathbf{0}\)-trader price \({S}^{\mathbf{0}}_t\) can be a steady state.} It has dynamic (\ref{ori}), making expected excess in any pre-disclosure trading window vanish for both trader factions.
\end{enumerate}
None above require selection pressure (\ref{mratio}): all p\&l-expectations are known pre-trade; under low RI, no optimising agents would fight against status quo benefit drifts
.

\begin{proposition}\label{p'}
   The market of Section \ref{canonbiased}, with heterogenous beliefs about model- or event-risk \(B\), has a FTAP-viable market-clearing price \(S^{het}_{t}=\langle{S}\rangle^{\hat{\pi}}_{t}\) ((\ref{Shetro})) such that during periods of low RI (Definition \ref{riilow}) it has an apparent Status Quo Bias, \(\{\hat{\pi}^\mathbf{1}_{t}\}\ll\{{\Pi}^\mathbf{1}_{t}\}\), relative to the fair-price process 
  \(\langle S\rangle^{{\Pi}} _t\) of the inferential participants.
\end{proposition}    

Even if \(RI_t>1\), outcome scenarios must be belief-weighted, requiring \(\frac{RI_t-1}{{\pi}^\mathbf{0}_t/{\pi}^\mathbf{1}_t}>1\) at least. If \(B\)-risk is subject to
ambiguity, the Inertia Axiom of Bewley (\cite{Bew}) applies: abandon status quo only if doing so dominates not doing so 
    (e.g. if \(\Delta S_t\sim0\) as \(RI_t\to\infty\)); this principle is even more sensible given Remark \ref{anythingatall}. To summarise, a market of heterogenous beliefs about \(B\)-risk can be trapped in a 
(quasi) \(B\)-sure state, despite \(B\)-data and \(B\)-inferential traders.

\begin{rem}\label{standard&biased}The \(\mathbf{0}\)-traders, aka. the \(\mathbf{0}\)-insiders (of Introduction), win, against the inferential \({\Pi}\)-
traders, aka. \(\mathbf{0}\)-outsiders. The 'trapped' market has nACMM in the inferential basis, thus offering arbitrage to the \(\mathbf{0}\)-outsiders. Yet it persists, as bets that pay only if \(B=\mathbf{1}\) and only at \(t=1\) are suboptimal except near \(t=1\). Viability thus favours \(\mathbf{0}\)-insiders, the \(B\)-sure, not \(\mathbf{0}\)-outsiders, the \(B\)-unsure: there is no dichotomy between 'viability' and 'efficiency'.
\end{rem}

\begin{rem}\label{stqb}The above is a rational market mechanism for Status Quo Bias, a 'behavioural bias' (\cite{sq1}, \cite{sq2}). It allows an intrinsically predictable risk known to many to be 'ignored'; such a market resembles, arbitrarily closely, one where outcomes arrive 'out of the blue' (causing it to 'gap'). Further, without public disclosure, despite being known to some participants, a risk can be 'missing' in market prices indefinitely; hence "a market can be irrational longer than you can be solvent".\end{rem}


\section{Summary and Discussion}\label{con-dis}
Viable asset-pricing under model- or event-risk is treated in a standard setting but for data that enable {pre-horizon} risk-resolutions
. It has intuitive and familiar forms (Proposition \ref{p3}, Corollary \ref{coro0}-\ref{coro0'} 
and Remark \ref{1-2}-\ref{ubi}) under common economic constraints.

Our risk-neutral equivalent approach ensures FTAP-viability, aka. NFLVR, known to be equivalent to being jointly NUPBR and NCA ('No Classical Arbitrage'). NFLVR prices in the inferential basis are only NUPBR in the physical basis (Proposition \ref{p4}, Appendix \ref{weakp})
: agents with foreknowledge have FLVR, but only in the form of classical arbitrage with neither {pre-resolution sure-payoff nor scalability} (due to the need for interim credits). For a predictable risk to a given 'status quo' however, the above means that it is {suboptimal} to long 'change' unless pro-'change' dataflow is sufficiently intense (Proposition \ref{p'}). Indeed, asset-prices that dismiss the 'change'-risk act as a \emph{stable attractor}, capable of delaying anticipation 'to the last minute'. A non-anticipative market, which 'gaps' upon risk-disclosure, has nACMM in the inferential basis and so offers arbitrage to inferential traders; yet it persists, as shown.

\begin{rem}\label{dis1}Were the above not so, no market can be viable against {any} doubt or fear. This mechanism 
is reinforced by Bewley's Inertia Axiom for risks that are ambiguous. Under common conditions then, 'viability' and 'efficiency', far from being dichotomous, are mutually reinforcing. However, this is at the expense of 'anticipation', a key goal of any market. By Section \ref{noantpn}, '{less continuity and/or less predictability}' necessitates anticipatory positioning; it weakens the very mechanism keeping 'the viability-efficiency dichotomy' at bay, however. How to find the right balance is a matter of great practical and theoretical relevance.
\end{rem}
\begin{appendix}
\appendix
\section*{APPENDIX}
\section{Properties of Binary Inferential Testing}\label{infps}
\begin{enumerate}

\item\label{takinglimit}\emph{Regular Continuous-Time Inference.} The sequential testing of given measure-pair \({W}
    _{\pm}\) relies on (log) likelihood-ratios (log-LR) \(l^{+
    }_t:=\log L^{+
    }_t:=\log\frac{d{W}
    _{+}|_t}{d{W}
    _{-}|_t}\), to which Bayes' Rule applies: 
    \(l^{+}_{s+t}=l^{{+}}_{s}+l^{+}_{t}\), \(\forall s,t\in\mathbb{R}^+\).
    For homogeneous data \(d{W}_B|_\tau=\mathbb{I}_{\{+\}}(B)r_\Delta d\tau +\sigma dw_\tau\), \(B\in\mathcal{B}:=\{+,-\}\), driven by standard Wiener noise \(w\) on \({\Omega}:=C_0[0,\infty)\), the log-LR process \(\{l^{+}_\tau\}\) satisfies (\cite{PShir}): with \emph{signal-to-noise} \(\sigma^l:=\frac{r_\Delta}{\sigma}\),\begin{equation}\label{w}dl^{+}_\tau(B)=(-1)^{\mathbb{I}_{\{-\}}(B)}\frac{(\sigma^l)^2}{2}d\tau + \sigma^l dw_\tau=\frac{dL^{+}_\tau}{{L^{+}_\tau}}(B)-\frac{(\sigma^l)^2}{2}d\tau.\end{equation} 

\item{\emph{Change of Clock.}\label{ckchng} A suitable \emph{time-change} (\cite{PShir}), \(
\tau\mapsto{t}
\), yields time-varying data, of continuous and uniformly bounded 
drift \(\{r_{{\Delta,t}}\}\) and volatility \(\{\sigma_t\}\) with respect to a time-\(t\) standard Wiener noise. Its test dynamic reads: with \(\sigma^l_t:=\frac{r_{{\Delta,t}}}{\sigma_t}\),
\begin{align}
    \label{w'}
&dl^{+}_t(B)=
(-1)^{\mathbb{I}_{\{-\}}(B)}\frac{(\sigma^{l}_t)^2}{2}dt+ \sigma^{l}_tdw_t=\frac{dL^{+}_t}{L^{+}_t}(B)-\frac{(\sigma^l_t)^2}{2}dt;
\\\label{rawdrift}  
&\frac{d\pi^+_{t}}{(\sigma^\pi_t)^{2}}(B)= \big(\mathbb{I}_{\{+\}}-\pi^+_t\big)(B)\cdot(\sigma^l_t)^2dt+\sigma^l_tdw_t,
\end{align}
where \(\{\pi^{\pm}_t\}\) is the associated \(B\)-belief, 
a martingale on \(\big(\mathcal{B}\times\Omega,\{\mathcal{F}_{t}\};\text{ }\pi^B_0\times W_B\big)\). 
}

\item\label{contim''}\emph{Finite-Time Resolution.} Standard setting resolves \(B\) at \(T_{\mathcal{B}}=\infty\). 
Finite times, \(T_{\mathcal{B}}=1
\), can be treated on a basis that is the image of the usual via a \emph{smooth monotone invertible map} \({\Phi}:[0,1)\to [0,\infty)\). Structures on \(C_0[0,\infty)\) are pulled-back to \(C_0[0,1)\) via: \(\Phi^*g:=g\circ\Phi,\text{ }g\in C_0[0,\infty)\), its \emph{uniform-convergence metric} included. The construction of Wiener measure \(w\) on \(C_0[0,\infty)\) (\cite{whitt}) then can be done on \(C_0[0,1)\), effectively pushed-forward: \((\Phi^*)_*dw(f):=dw\circ(\Phi^*)^{-1}(f),\text{ }f\in C_0[0,1)\). 
We can be more explicit. Consider the series of Wiener spaces \(C_0[0,n]\), and its image, \(C_0[0,1_{n-}]\), \(1_{n-}:=\Phi^{-1}[n]<1\), \(n\in\mathbb{N}\). All are \(C_0[0,1]\) up to normalisation; each \({\Phi}|_{[0,1_{n-}]}\) is a smooth and deterministic time-change. On \(C_0[0,1_{n-}]\), \(n\) fixed, a time-\(\tau\) standard noise \(\{dw_\tau\}|_n\) becomes \(\{\sqrt{\Phi'_t}\cdot{dw_t}\}|_{1_{n-}}\) 
in time-\(t\) standard noise. 
Over \([0,1_{n-}]\), data dynamic is of form (\ref{bD}), and inference, (\ref{w'}-\ref{rawdrift}), with drift and squared-volatility diverging in line with \(\Phi'_t\) as \(t\to1\). 

\item{\label{convergent}\emph{Inferential Resolution.} We have \(|l^{+}_t|
<\infty\) almost surely where equivalence \({{W}}_{+}|_t\sim{W}_{-}|_t\) holds (Radon-Nikodym Theorem); so for Wiener-drift tests (\ref{w}-\ref{w'}): almost surely,
\begin{equation}\label{resovlingls'}\int_0^t(\sigma^l_u)^2du\equiv \int_0^{t}\big(\frac{r_{\Delta,u}}{\sigma_u}\big)^2du<\infty.
\end{equation}
There is \(B\)-resolution, \(\lim_{t\to T_{\mathcal{B}}}|l^{+}_t|\to\infty\), 
\emph{iff.} under \({{W}}_{+}\perp{W}_{-}\)
: almost surely, 
\begin{equation}\label{resovlingls}\lim_{t\to T_{\mathcal{B}}}\int_0^{t}(\sigma^l_u)^2du\equiv\lim_{t\to T_{\mathcal{B}}}\int_0^{t}\big(\frac{r_{\Delta,u}}{\sigma_u}\big)^2du=\infty.
\end{equation}Risk-resolution thus 
precludes martingale measures for the log-LR process\footnote{\label{pisignaltonoise}This is obvious for belief process \(\{\pi^{\overline{b}}_t(b)\}\in[0,1]\) on \([0,T_{\mathcal{B}}]\) as \(W_{b}(\text{ }\{\text{ }\pi^{\overline{b}}_{T_{\mathcal{B}}}(b)=0\text{ }\}\text{ })=1\). Nevertheless, note \(\frac{1}{2}E\int_0^{T_{\mathcal{B}}}(\pi^{\overline{b}}_u(b)\cdot\sigma^l_u)^2du<\infty\), by \(\frac{1}{2}(\pi^{\overline{b}}_u(b)\cdot\sigma^l_u)^2
=d\log[\pi^b_u(b)]-\pi^{\overline{b}}_u(b)\cdot\sigma^l_udw_u\) ((\ref{rawdrift})). A similar finite 'squared integral' not yielding a martingale measure is mentioned in \cite{ankimk} (after Theorem 2.10).
} \(\{l^{+}_t(b)\}\).
}
\item{\label{tails}\emph{Adjacency.} Resolving and so divergent log-LR processes \(\{l^{+}_t\}\) and \(\{\hat{l}^{+}_t\}\), with respective test measure-pairs \({W}_\pm\) and \(\hat{{W}}_\pm\) for which equivalence \({{W}}_{\pm}\sim \hat{W}_\pm\) holds, 
can differ at most by a finite amount almost surely (Radon-Nikodym Theorem): at any \(t<T_{\mathcal{B}}\), to which the restriction of all four measures are equivalent, we have,\begin{equation}\label{logtrick}\big|\hat{l}^{+}_t-l^{+}_t\big|=\bigg|\log{\frac{d\hat{{W}}_{+}|_t}{d{W}_+|_t}}-\log{\frac{d\hat{{W}}_{-}|_t}{d{W}_{-|_t}}}\bigg|;\end{equation}the RHS is almost surely finite as \(t\to T_{\mathcal{B}}\), by \({{W}}_{\pm}\sim \hat{W}_\pm\). Write \(\{\hat{l}^{+}_t\}\approx\{l^{+}_t\}\) for adjacency.
}

\item\emph{Wiener Adjacency and Equivalence.}\label{tailsame} Consider Wiener test measure-pair \(\hat{{W}}_\pm\) of the form \(d\hat{{W}}_\pm|_t=-\frac{1}{2}\hat{\theta}_{\pm,t}dt+d{W}_\pm|_t\) relative to the original \({W}_\pm\), so that \(\hat{r}_{\Delta,t}={r}_{\Delta,t}-\frac{1}{2}\hat{\theta}_{\Delta,t}\). Its log-LR dynamic \(d\hat{l}^{+}_t\) has form (\ref{w'}), in terms of \(d\hat{w}_t:=\frac{d\hat{W}_-|_t}{\sigma_t}\) and \(\hat{\sigma}^l_t:=\frac{\hat{r}_{\Delta,t}}{\sigma_t}\). 
 If the measure generated by \(\{d\hat{l}^{+}_u\}|_t\) 
is equivalent to that by \(\{d{l}^{+}_u\}|_t\), write \(\{\hat{l}^{+}_u\}|_t\sim\{{l}^{+}_u\}|_t\) and call them equivalent. Such equivalence in the Wiener setting demands \(\hat{\sigma}^l_u=\sigma^l_u\), that is, \(\hat{\theta}_{\Delta,u}=0
 \) and \(d\hat{{W}}_{\pm}|_t=-\frac{1}{2}\hat{\theta}_{t}dt+d {W}_\pm|_t\) for some \(B\)-{independent} drift \(\hat{\theta}_{t}\). Then, by \((\frac{\hat{\theta}_t}{r_{{\Delta},t}}\sigma^l_t)^2\equiv(\frac{\hat{\theta}_t}{\sigma_t})^2\),
 \begin{equation}\label{lhat=l+drft}E \big|\hat{l}^{+}_t-l^{+}_t\big|=
 \frac{1}{2}E\int_0^t\frac{|\hat{\theta}_u|}{r_{\Delta,u}}(\sigma^l_u)^2du\equiv
 \frac{1}{2}E\int^t_0\frac{|\hat{\theta}_u|\cdot r_{\Delta,u}}{(\sigma_u)^2}du,\end{equation}with 
\(E\) denoting expectation 
under \emph{either} law \(W_\pm\); so Novikov's Condition is met by the relative drift 
of \(\{\hat{l}^{+}_u\}|_t\) \emph{vs} \(\{l^{+}_u\}|_t\) \emph{iff.} it is met by that 
of \(\hat{W}_\pm|_t\) \emph{vs} \(W_\pm|_t\). Thus, the following relationships are equivalent over any \([0,t]\), \(t<T_{\mathcal{B}}\),\begin{equation}\label{neat}\hat{W}_\pm|_t\sim W_\pm|_t\text{ }\Longleftrightarrow\text{ }\{\hat{l}^{+}_u(\pm)\}|_t\approx\{{l}^{+}_u(\pm)\}|_t\text{ }\Longleftrightarrow\text{ }\{\hat{l}^{+}_u(\pm)\}|_t\sim\{{l}^{+}_u(\pm)\}|_t.\end{equation}

\item{\label{unique}\emph{Extension to Resolution-Time \(T_{\mathcal{B}}\).} 
Recall that no resolving log-LR processes can have martingale measures as \(t\to T_{\mathcal{B}}\) (Item-\ref{convergent}); so the last part of (\ref{neat}) has no meaning into resolution \(T_{\mathcal{B}}\). 
Adjacency into resolution however, 
by (\ref{lhat=l+drft}), does, provided:\begin{align}\label{finitereldrft}&\lim_{t\to T_{\mathcal{B}}}E\int^{t}_0\frac{|\hat{\theta}_u|}{r_{\Delta,u}}(\sigma^l_u)^2du=\lim_{t\to T_{\mathcal{B}}}E\int^{t}_0\frac{|\hat{\theta}_u|\cdot r_{\Delta,u}}{(\sigma_u)^2}du<\infty,\text{ and so,}\\\label{novs}&\lim_{t\to T_{\mathcal{B}}}E\int^{t}_0\big(\frac{\hat{\theta}_u}{r_{\Delta,u}}\big)^2(\sigma^l_u)^2du=\lim_{t\to T_{\mathcal{B}}}E\int^{t}_0\big(\frac{\hat{\theta}_u}{\sigma_u}\big)^2du<\infty.\end{align}This is sufficient for asset-pricing (Appendix \ref{weakp}), also trivially true for log-LR pairs arisen from equivalent test measure-pairs \(\hat{W}_\pm\sim W_\pm\) 
defined on \(C_0[0,T_{\mathcal{B}})\) ((\ref{logtrick})).
}

\end{enumerate}

\section{Proposition \ref{p3}, Corollary \ref{coro0}-\ref{coro0'} and NUPBR-Viability}\label{proofp3} 
\subsection{FTAP-Viable Pricing in the Inferential Basis}\label{proofp3.1} It suffices to prove it for \(B\)-risk only pricing: \(\langle{Y}\rangle^A_t:=S_t+\check{RP}(t)=Y^-_t+X_\Delta(t)\cdot A^+_t\) ((\ref{S0A})). Note that the expected benefit process \(Y_t\equiv\langle{Y}\rangle^\pi_t=Y^-_t+X_\Delta(t)\cdot \pi^+_t\) ((\ref{yt})) is already a martingale of the inferential basis \(\big(\mathcal{B}\times{\mathbf{\Omega}}_{1},\{\mathbf{F}_t\};\text{ }\pi^B_0\mathbf{W}_B\big)\), with \(\mathbf{{\Omega}}^{1}:=C_0[0,1)\times C_0[0,1]\) and \(\{\mathbf{F}_t\}:=\{\mathcal{F}^D_t\otimes\mathcal{F}_t\}\), given \(B\)-conditionally independent joint-data \(\{(D_t, Z_t)\}\) under law \(\pi^{B}_0{\mathbf{W}}_{B}:=\pi^{B}_0(W^D_B\times W_B|_{1})\).

FTAP-viability in the Wiener setting requires \(d\langle{Y}\rangle^\pi_t\) and \(d\langle{Y}\rangle^A_t\) to differ only by drifts. Any RNE measure \(\hat{\pi}^B_0\hat{\mathbf{W}}_{B}\) thus must preserve the \(B\)-conditional independence above, and so, for \(B\)-risk only pricing, must have the form \(\hat{\mathbf{W}}_{B}=\hat{W}^D_B\times {W}_B|_{1}\), \(\hat{W}^D_B\sim{W}^D_B\).

For term matching between 
  \(dA^+_t\) and \(d\pi^+_t\)
  , it is easier to work with \(\hat{l}^{+}_t:=\log\frac{A^+_t}{A^-_t}=:\log \hat{L}^{+}_t\), well-defined given \(0<\{A^+_t\}<1\) on any \([0,1-\zeta]\), \(\zeta\ll1\)
  . The dynamics of \(d\hat{l}^{+}_t\) and \(dA^+_t\) have the form and relationship of (\ref{w'}) and (\ref{rawdrift}) respectively. As the reference log-LR process is a sum \({l}^{+}_t={l}^{D{+}}_t+{l}^{Z{+}}_t\) of two independent ones adapted to \(\{\mathcal{F}^D_t\}\) and \(\{\mathcal{F}_t\}\) respectively, 
we have, given viability, \(\hat{l}^{+}_t=\hat{l}^{D{+}}_t+\text{ }{l}^{Z{+}}_t\), with the second part untouched by \(B\)-risk only pricing, 
and the first part equivalent to \({l}^{D+}_t\) on \([0,1-\zeta]\), 
that is: 
\(\exists\) \(B\)-independent drift \(\{\hat{\theta}^{D}_t\}\), \(t\in[0,1-\zeta]\), such that 
the following hold and remain so as \(\zeta\to0\) (
(\ref{neat}-\ref{novs})):\begin{align}\label{findRNE}
d\hat{W}^{D}_\pm|_t&=-\frac{1}{2}\hat{\theta}^{D}_tdt+dW^D_\pm|_t\sim dW^D_\pm|_t,\\    
    \label{w''}
d\hat{l}^{D\frac{+}{-}}_t(B)&=\frac{1}{2}(-1)^{\mathbb{I}_{\{-\}}(B)}(\sigma^{lD}_t)^2dt+\sigma^{lD}_td\hat{w}^D_t
\end{align}with \(d\hat{w}^D_t:=\frac{d\hat{W}^{D}_-}{\sigma^D_t}\). 
Proposition \ref{p3} follows, as does Corollary \ref{coro0'}. \(\blacksquare\)

\subsection{Canonical FTAP-Viable Pricing in the Inferential Basis}\label{coro1proof}

Under Proposition \ref{p3}, condition (\ref{kcon}) demands \(\hat{l}^{+}_t(\pm)= g[{l}^{+}_t(\pm)]\) where \(g\) is 2-differentiable (and time-homogeneous). Both \(d{l}^{+}_t\) and \(d\hat{l}^{+}_t\) follow the log-LR dynamic (\ref{w'}). 
Under FTAP-viability, only their drifts may differ, which requires \(g'=1\) (due to the \(g'\cdot d{l}^{+}_t\) term in Ito's Formula); 
hence Corollary \ref{coro0}.

In the context of Appendix \ref{proofp3.1} (\(B\)-risk only pricing), by integrating (\ref{w''}) we have the following gauge for the 'RNE distortion' of the reference belief process up to any \(t\le1\):\begin{equation}\label{corolaw}
\hat{\kappa}_t:=\bigg(\frac{L^{+}_t}{L^{+}_0}\bigg)\bigg/\bigg(\frac{\hat{L}^{+}_t}{\hat{L}^{+}_0}\bigg)=\text{ }\exp
{\frac{1}{2}\int_0^{t}\frac{\hat{\theta}^{D}_{u}}{r^{D}_{\Delta,u}}\cdot(\sigma^{lD}_u)^2du}
<\infty.\end{equation}
As canonical pricing means \(\hat{\theta}^{D}_{u}\equiv0\), we have \(\hat{\kappa}_{t}\equiv1\) and so Corollary \ref{coro1}. \(\blacksquare\)

\subsection{Minimal Viability in the Realised Stochastic Basis}\label{weakp}

We proceed under (\ref{assmption3}-\ref{fixedimp}) for brevity, without affecting the conclusion, nor the argument\footnote{\label{sic'}
    Without (\ref{assmption3}-\ref{fixedimp}), inference depends on benefit data \(\{Z_t\}\) as well as on \(\{D_t\}\), any inferential contribution from \(\{Z_t\}\), controlled by \(\int_0^{t}
(\sigma^{l}_u)^2du\), is finite throughout and so irrelevant to the viability question.}. Given 
FTAP-viable dynamic (\ref{fulldynamics}) in inferential basis \(\big(\mathcal{B}\times \mathbf{\Omega}_{1},\{\mathbf{F}_t\};\text{ }{\pi}^B_0({W}^D_B\times{W}_B|_{1})\big)\) with RNE law \(\hat{\pi}^B_0(\hat{W}^D_B\times\check{W}^T_B|_{1}) \sim{\pi}^B_0({W}^D_B\times{W}_B|_{1})\), we study its viability under \(\{B=b\}\) in the outcome-basis \(\big(\mathbf{\Omega}_{1},\{\mathbf{F}_t\};\text{ }{W}^{D}_b\times W_b|_{1}
\big)\). The latter is also the basis of an otherwise identical trader who 'knows' \(B=b\): our question is one of \emph{initially enlarged filtrations} by \(\{B=b\}\).
\begin{proposition}\label{p4}
  If an asset-price process under model- or event-risk \(B\in\{+,-\}\) is FTAP-viable in the inferential basis, in which the risk-belief process \(\{\pi^\pm_t\}\) (equivalently, \(\{{L}^{+}_t\}\) or \(\{{l}^{+}_t\}\)) is continuous throughout, then it is only NUPBR in the realised, true-outcome, basis.  
\end{proposition}
\begin{proof}
It is a direct take of Part-(1) of Theorem 1.12 of \cite{fonz3} (or Theorem 3.2 of \cite{sideinfo2}). We show it explicitly in our context, with applications in mind. 

On \([0,1]\), with RNE law \(\hat{\pi}^B_0(\hat{W}^D_B\times\check{W}^T_B|_{1})\), asset-price \(S_t=S^-_t+X_\Delta\cdot\hat{\pi}^+_t\) 
obeys dynamic (\ref{fulldynamics}). 
For the RNE and inferential belief processes, adjacency (Item-\ref{tails}, Appendix \ref{infps}) 
implies:\begin{equation}\label{orders}|\pi^+_t-\hat{\pi}^+_t|=\mathcal{O}((\sigma^\pi_t)^2).
\end{equation}The price dynamic in the realised basis under say \(\{B=-\}\) is given by setting \(B=-\) in (\ref{fulldynamics}) and re-writing it in the original standard Wiener noise \(\{d{w}^{D}_{t}\}\). By (\ref{findRNE}) then, price drifts due to \(B\)-inference, 
 written as price-of-diffusion-risk, read:\[-\big[\frac{1}{2}\frac{\hat{\theta}^D_t}{\sigma^D_t}+\hat{\pi}^+_t(-)\cdot{\sigma^{lD}_t}\big];
 \]both 
 finitely square-integrable, 
 the first by (\ref{finitereldrft}-\ref{novs}), 
the second, (\ref{orders}) and footnote-\ref{pisignaltonoise}, thus meeting the 
condition for NUPBR (Corollary 3.3.15, \cite{fonz1}). 
\end{proof}

   

\subsection{Biased \(B\)-sure Beliefs and Bounded Total Bias}\label{weakp'}

As shown in Appendix \ref{coro1proof}, only canonical asset-prices can satisfy the economic risk-pricing demand of (\ref{kcon}). Markets, however, may not 'heed' this for a \(B\)-risk that can be hedged and so priced by arbitrage; so non-canonical price dynamic (\ref{fulldynamics}) cannot be dismissed.

\begin{enumerate}

 \item \emph{Canonical pricing perturbed by noise.}\label{noisycan} 
 Any given viable price process has at any instant \(t\le1\) an implied \(B\)-risk pricing parameter \(\hat{K}_t\) via (\ref{corolaw}): with \(K:=\frac{L^{{+}}_0}{\hat{L}^{{+}}_0}\),\begin{equation}\label{cumulativeeffect}
     \hat{K}_t:=
     \frac{L^{{+}}_t}{\hat{L}^{{+}}_t}=\hat{\kappa}_t\times K 
     .
 \end{equation}If \(\{\hat{\kappa }_{t}\}\approx1\) and \(\{E\text{ }\hat{\kappa }_{t}\}\approx1\), such pricing is essentially canonical but 'noisy'.
 
 \item\emph{Canonical Pricing with biased \(B\)-sure beliefs.}\label{biasedcan} If \(\{\hat{\kappa}_t\}\) deviates systematically from \(1\), 
 it is more plausibly due to biased reference \(B\)-sure beliefs (e.g. of the form (\ref{findRNE})) than the violation of basic principles (e.g. (\ref{kcon})). 
Viability requires such bias to be \(B\)-independent (Item-\ref{tailsame}, Appendix \ref{infps}), and its cumulative effect, bounded ((\ref{corolaw}))
.

\end{enumerate}

\end{appendix}


\end{document}